\documentclass[runningheads]{llncs}
\usepackage{amssymb,amsmath}
\usepackage{verbatim}
\usepackage{epsfig}
\usepackage{algorithmicx,algorithm,algpseudocode}

\usepackage[utf8]{inputenc}
\usepackage[T1]{fontenc} % TEMP
\DeclareSymbolFont{rsfscript}{OMS}{rsfs}{m}{n}
\DeclareSymbolFontAlphabet{\mathrsfs}{rsfscript}

\spnewtheorem{fact}{Fact}{\bfseries}{\itshape}

\newcommand{\keywords}[1]{\par\addvspace\baselineskip
\noindent\keywordname\enspace\ignorespaces#1}

\usepackage{geometry,array,graphicx}
\usepackage{soul}
\usepackage{color}
\usepackage[usenames,dvipsnames]{xcolor}
\newif\ifdebug

\definecolor{note}{RGB}{70,70,70}

\sethlcolor{Goldenrod}

\newcommand*\samethanks[1][\value{footnote}]{\footnotemark[#1]}

\begin{document}
\title{A New Heuristic Synchronizing Algorithm}
\author{Jakub Kowalski\thanks{Supported in part by Polish MNiSZW grant IP2012 052272}, Marek Szyku{\l}a\samethanks}
\institute{Department of Mathematics and Computer Science, University of Wroc{\l}aw \\
\email{\{kot,msz\}@ii.uni.wroc.pl}}
\maketitle

\begin{abstract}
We present a new heuristic algorithm finding reset words. The algorithm called CutOff-IBFS is based on a simple idea of inverse breadth-first-search in the power automaton. We perform an experimental investigation of effectiveness compared to other algorithms existing in literature, which yields that our method generally finds a shorter word in an average case and works well in practice.

\end{abstract}

\section{Introduction}
Synchronizing automata are important in various fields, such as model-based testing \cite{BJKLP2005}, robotics (for designing so-called part orienters) \cite{AV2003}, bioinformatics (the reset problem) \cite{BAPLS2003}, network theory \cite{Ka2002}, theory of codes \cite{Ju2008} etc. In many applications it is important to find a reset word as short as possible. Unfortunately the problem of finding a shortest reset word is shown to be $\mathrm{FP^{NP[log]}}$-hard, and the related decision problem is both NP- and coNP-hard \cite{OM2010} (cf. also \cite{Be2010} and \cite{Ma2009,Ma2011} for approximation hardness and special classes).

Nevertheless there exist many exponential algorithms to deal with this problem \cite{KRW2012,RS1993,Sa2005,ST2011,Tr2006} and a lot of polynomial heuristics finding relatively the shortest reset words \cite{GH2011,KRW2012,Ro2009,Roman2009,Tr2006}. Recently we developed an exact algorithm based on a bidirectional-breadth-first search, which is currently the fastest algorithm for the problem \cite{KKS2013}. Our new heuristic algorithm was also used as a part of our exact algorithm.

The algorithm, called CutOff-IBFS, is a heuristic polynomial synchronizing algorithm finding a synchronizing word within the given length. We analyzed some properties of the algorithm and performed an experimental investigation of the algorithm and compared with some others existing in literature. Our method, although simple, seems to be more efficient in practice and generally finds shorter reset words in less time. For simplicity, we describe only the version which find the length of the found word, and later consider a suitable modification to find the word itself.

\section{Description of the Algorithm}
Before we run the main part of the algorithm we must run some other algorithm to obtain a reset word and an upper bound on the length of the shortest reset words. This is necessary as we do not have a guarantee that it finds any word in all cases. To have this property it can be easily combined with some other algorithm. We use well-known Eppstein algorithm \cite{Ep1990} as a preceding algorithm. Then we apply the length of the found word into CutOff-IBFS as $\mathtt{maxlen}$. It is preferred to use a fast preceding algorithm with not greater complexity than CutOff-IBFS, to not increase the overall complexity.

The formal description is given in Algorithm~\ref{alg_cutoffibfs}. We start from the list $L$, which contains all singletons of the states. Then in each step we create the set of preimages of the sets from $L$ by each letter. The sets are stored in the trie $T_{ic}$, which allows us to exclude all duplicated sets. Next we create the resulted list $L'$ by taking only the $\mathtt{maxsize}$ largest sets from $T_{ic}$. We repeat the steps until we obtain the complete set or we have done $\mathtt{maxlen}$. In the second case we use the word found by the preceding algorithm as a result.

\begin{algorithm}\caption{CutOff-IBFS}\label{alg_cutoffibfs}
\begin{algorithmic}[1]
\Require $A = \langle Q,\Sigma,\delta \rangle$ -- an automaton with $n=|Q|$ states and $k =|\Sigma|$ input letters.
\Require $\mathtt{maxlen}$ -- maximum length of words to be checked.
\Require $\mathtt{maxsize}$ -- maximum size of the lists.
\Procedure{CutOff-IBFS}{$A,\mathtt{maxlen},\mathtt{maxsize}$}
\State $L \gets$ \Call{EmptyList}{}
\ForAll{$q \in Q$}
  \State \Call{$L$.insert}{$\{q\}$}
\EndFor
\For{$l \gets 1$ \textbf{to} $\mathtt{maxlen}$}
  \State $T_{ic} \gets$ \Call{EmptyTrie}{} \Comment{A trie to store next sets}
  \ForAll{$S \in L$}
    \ForAll{$a \in \Sigma$}
      \State $S' = \delta(S,a)$ \Comment{Compute the image of $S$ by $a$}
      \If{$|S'|=n$} \Comment{The goal test}
        \State \Return $l$
      \EndIf    
      \State \Call{$T_{ic}$.insert}{$S'$}
    \EndFor
  \EndFor
  \State $L' \gets$ \Call{EmptyList}{}
  \State \Call{$L'$.insert}{$T_{ic}$}
  \State Sort $L'$ by descending cardinalities.
  \State Trim $L'$ to $\mathtt{maxsize}$ length, by truncating from the end.
  \State $L \gets L'$
\EndFor
\State \Return ``Not found a synchronizing word of length $\leq \mathtt{maxlen}$''
\EndProcedure
\end{algorithmic}
\end{algorithm}

\subsection{Analysis}
The correctness of the algorithm comes from the fact that we can assign a word to each computed set, so that an image of a set by its assigned word is a singleton. The word assigned to the final set $S'$ of size $n$ would be a synchronizing one.

Complexity of the algorithm depends of the given $\mathtt{maxsize}$ parameter. The larger $\mathtt{maxsize}$ is the slower the algorithm works and uses more memory, but it finds shorter words. Complexity of the algorithm depends linearly on $c$ so it is easy to set a desirable trade-off. Also it works fast when it finds a short word, and this is the case for random automata. A trie is used to skip sets already stored in it, keeping insertion time in $O(n)$.

\begin{theorem}CutOff-IBFS works in $O(lckn+kn^2)$, where $l$ is the length of the found reset word and $c=\mathtt{maxsize} \ge 1$. The space complexity is $O(ckn)$.\end{theorem}
\begin{proof}
Initialization in lines~2--5 takes $O(n^2)$ time ($n$ sets). We assume now the worst case, that the word is not found and the \textbf{for} loop in line~6 runs $l$ times. During all the executions except the first we have at most $c$ sets in $L$. Computing preimages of the sets (line~10) by each of the letter can be done in $O(ckn)$ time ($O(n)$ for each preimage). Insertion the preimages into a trie (line~14) is done in $O(ckn)$ time ($O(n)$ for an insertion). Sorting $L$ (line~19) can be done in $O(ckn)$ time by counting sort. During the first execution of the \textbf{for} loop there are $n$ sets at the beginning, so it takes $O(kn^2)$ time. The total time complexity is $O(lckn+kn^2)$. Space is required to store the list $L$ and the trie $T_c$ during all steps, which yields $O(ckn)$.
\end{proof}

Together with preceding Eppstein algorithm (see \cite{Ep1990}, it works in time $O(kn^2+n^3)$ and $O(kn^2)$ space), and taking $l \in O(n^3)$, the algorithm has $O(ckn^4)$ time complexity. The total space complexity is $O(ckn+kn^2)$.

Inverse breadth-first-search is very effective for some of the most extremal automata (with the longest reset lengths). Similar as our exact algorithm \cite{KKS2013} works in polynomial time for them, our heuristic algorithm always find the shortest reset word in these cases.

\begin{theorem}\label{th_cutoffibfs_slowlysynchro}
When $\mathtt{maxsize} \ge n$, Algorithm~\ref{alg_cutoffibfs} always finds the shortest reset word for the \v{C}ern\'{y} automaton $\mathrsfs{C}_n$ and the slowly synchronizing series introduced in \textrm{\cite{AGV2012}} $\mathrsfs{D'}_n$,$\mathrsfs{W}_n$,$\mathrsfs{F}_n$,$\mathrsfs{E}_n$,$\mathrsfs{D''}_n$,$\mathrsfs{B}_n$,$\mathrsfs{G}_n$,$\mathrsfs{H}_n$.
\end{theorem}
\begin{proof}
Consider the automaton $\mathrsfs{C}_n$. We will show that a set corresponding with the shortest reset word is kept in the list during all executions of the main loop in line~6, that is after $i$ steps there will be a set $S$ which is a preimage of some singleton by the word consisted of the $i$ last letters of the shortest reset word. This is true at the beginning since we have all the singletons in the list.

After the first step (the first execution of the main loop in line~7) there will be a set of size $2$ in the list and the singletons, since it is kept because of the sorting in line~19. After each of the $(n-1)$ next steps there is added a new set of size $2$ and all of them are kept since there are only $n$ such sets after the $n$-th step. In the $(n+1)$-th step 
a set of size $3$ is created and it is the largest set in the list. We can continue in this way and after each of part of $(n-1)$ steps there will be a step introducing a new set of size greater by $1$, while within a part each step introduces exactly one set of the currently largest size. There are $(n-1)$ steps introducing a larger set and $(n-2)$ parts consisted of $(n-2)$ steps. So in the $((n-1)^2)$-th step there will be introduced the final set of size $n$.

In a similar way one may follow execution of the algorithm for the other automata series to see that a set corresponding with the shortest reset word is kept in the list, and so the algorithm finds the shortest reset word.
\end{proof}

\subsection{Finding a Reset Word}
The algorithm can also return the found word, not only its length. It however increases space complexity, but the time complexity can be kept. We cannot store complete words together with the sets since they may be long ($l$), so creating a new set would take $O(n+l)$ time due to copying issue. Instead we can store, together with a set, the applied letter and a pointer to the information stored for the preceding set. This additional information costs a constant space for a set and can be computed in a constant time for a preimage set. However we need to preserve this information for each constructed set, so it will take additional $O(cl+n)$ space. This would yield total space complexity $O(c(l+kn)+kn^2)$ space complexity, while keeping time complexity in $O(lckn+kn^2)$.

\subsection{Technical Improvements}
We discuss here some technical improvements of the algorithm. They can reduce running time and space but within a constant factor.

First of all, we can skip sorting from the main loop. Instead we can maintain $n$ tries for each possible cardinality of a set. Then we can construct a list by taking sets from these tries from the largest to the smallest.

Another simple improvement comes from the fact, that we can start only from the singletons of states from the sink component, as well as only from the states having in-degree at least $2$ on some letter.

One may also permute the automaton before running the algorithm. Since the sets during computation are usually small, it could be better to have states occurring more frequently to be at the top of tries. It would lead to have smaller heights of the tries and so to faster execution. A simple heuristic method to do that is sorting the states decreasing by their in-degrees.

\section{Empirical Behavior}

\begin{figure}[!ht]
 \centering
 \epsfig{file=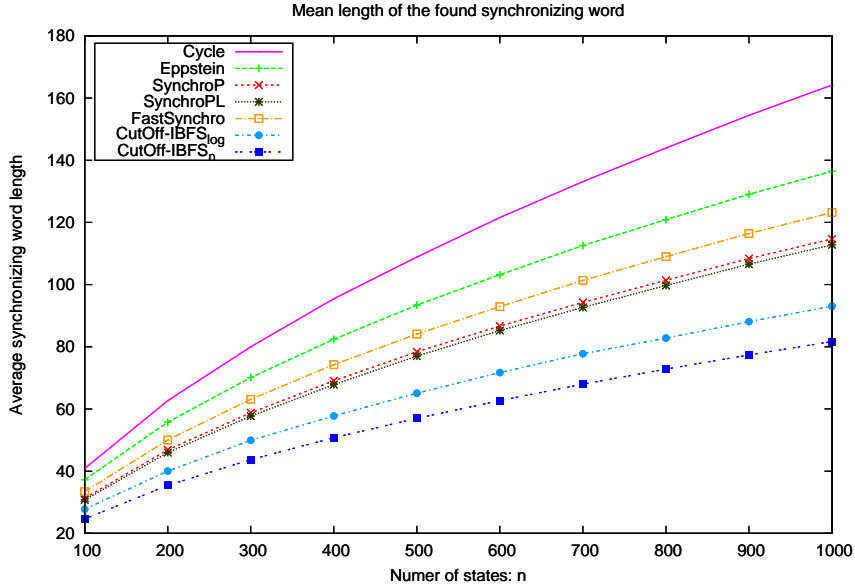, width=80mm, angle=-90}
 \caption[]{The mean found length by the algorithms for random automata.}
 \label{fig:heur_mean}
\end{figure}

\begin{figure}[!ht]
 \centering
 \epsfig{file=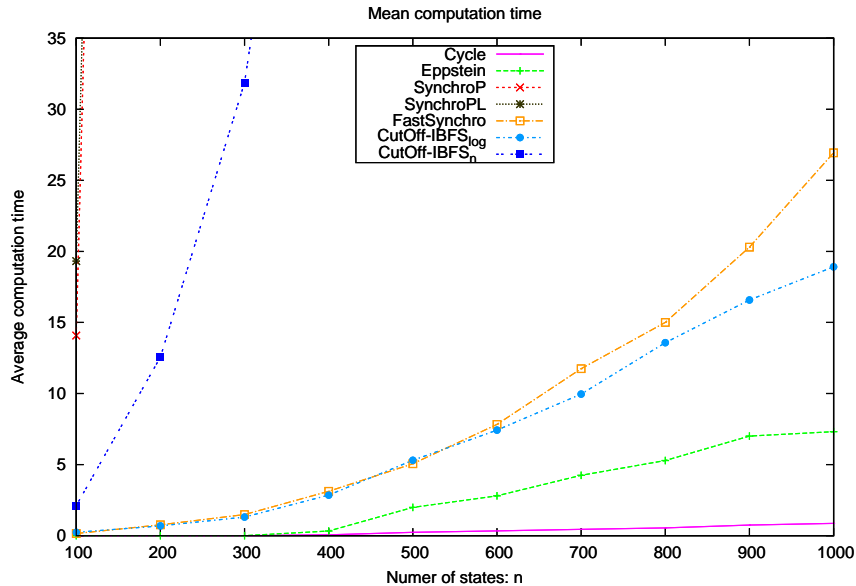, width=80mm, angle=-90}
 \caption[]{The average execution time of the algorithms for random automata.}
 \label{fig:heur_time}
\end{figure}

We have compared some known synchronizing algorithms in terms of the quality (found word lengths) and the execution time. The experiment was performed on $10,000$ uniformly random labeled automata for $n=100,200,\ldots,1000$ (each algorithm worked on the same automata set). Figure \ref{fig:heur_mean} shows the mean length found by the algorithms which is generally the quality measure. Figure \ref{fig:heur_time} shows the average execution time used by the algorithms.

We presented two versions of CutOff-IBFS, with $c=\log n$ and $c=n$. See \cite{Roman2009,KRW2012} for the other algorithms. We can see that CutOff-IBFS found shorter reset words than the other algorithms and the $\log n$ version is slower only than two fastest Eppstein and Cycle algorithms.

Since there may be different versions of these algorithms or they may be run with different parameters, which can affect the results of the experiment, we describe here some implementation details. For CutOff-IBFS we included time used by Eppstein algorithm. We used the greedy versions of Eppstein algorithm \cite{Ep1990} and Cycle algorithm \cite{Roman2009,KRW2012}, which always selects a pair of states with the shortest synchronizing word. For the tree of the pair automaton we used $n$ as the constant used for marking nodes (this is $y$ constant in Algorithm~2 from \cite{Ep1990}).

%%%%%%%%%%%%%%%%%%%%%%%%%%%%%%%%%%%%%%%%%%%%%%%%%%%%%%%%%%%%%%%%%%%%%%%%%%%%%%%%%
\bibliographystyle{plain}
\bibliography{bibliography}
\end{document}